\documentclass[letterpaper, 10 pt, conference]{ieeeconf}  
\IEEEoverridecommandlockouts             
\usepackage{amsmath, amssymb, amsfonts}
\usepackage[colorlinks=true, citecolor=blue, linkcolor=blue]{hyperref}
\let\labelindent\relax
\usepackage{url, footnote, enumitem}
\usepackage{tikz}\usetikzlibrary{calc, positioning}
\usepackage{graphicx, cite, textcomp, xcolor}

\usepackage{amsmath, amsfonts, amssymb, amsthm}
\allowdisplaybreaks
\theoremstyle{plain}

\newtheorem{assumption}{Assumption}
\newtheorem{definition}{Definition}
\newtheorem{lemma}{Lemma}
\newtheorem{proposition}{Proposition}
\newtheorem{corollary}{Corollary}
\newtheorem{theorem}{Theorem}

\newtheorem{fact}{Fact} 
\newtheorem{example}{Example}

\title{\LARGE \bf Koopman Operator in the Weighted Function Spaces and its Learning for the Estimation of Lyapunov and Zubov Functions*}
\author{Wentao Tang$^{1}$
\thanks{*This work is supported by the faculty startup fund from the NC State University and ACS Petroleum Research Fund \#66911-DNI9.}
\thanks{$^{1}$ Wentao Tang is an assistant professor with the Department of Chemical and Biomolecular Engineering, North Carolina State University, Raleigh, NC 27695, U.S.A. {\tt\small wentao\_tang@ncsu}}
}

\begin{document}
\maketitle\thispagestyle{empty}\pagestyle{empty}

\begin{abstract}
The mathematical properties and data-driven learning of the Koopman operator, which represents nonlinear dynamics as a linear mapping on a properly defined functional spaces, have become key problems in nonlinear system identification and control. However, Koopman operators that are approximately learned from snapshot data may not always accurately predict the system evolution on long horizons. 
In this work, by defining the Koopman operator on a space of \emph{weighted} continuous functions and learning it on a \emph{weighted} reproducing kernel Hilbert space, the Koopman operator is guaranteed to be contractive and the accumulation learning error is bounded. The weighting function, assumed to be known \textit{a priori}, has an exponential decay with the flow or decays exponentially when compensated by an exponential factor. 
Under such a construction, the Koopman operator learned from data is used to estimate (i) Lyapunov functions for globally asymptotically stable dynamics, and (ii) Zubov-Lyapunov functions that characterize the domain of attraction. For these estimations, probabilistic bounds on the errors are derived. 
\end{abstract}

\section{\textsc{Introduction}}
Koopman operator of a dynamical system, defined on a suitable function space as the mapping from any function (called \emph{observable}) to its composition with the flow of the dynamics, provides a useful tool to identify, analyze, and control nonlinear systems \cite{mauroy2020koopman, bevanda2021koopman, brunton2022modern}. The literature in the recent decade has shown the wide application of this concept in fluid dynamics \cite{mezic2013analysis, arbabi2017study}, robotics \cite{bruder2020data, shi2022deep}, and chemical processes \cite{son2021application, zhang2023reduced}. 
While the Koopman operator is in principle defined on an infinite-dimensional function space and difficult to characterize (due to its non-closed, non-compact properties) \cite{singh1993composition}, with large data, its finite-rank approximations can be constructed via machine learning algorithms \cite{williams2015data, klus2016numerical}. Such approximations have been proven to possess desirable properties, such as the convergence in operator topology in the large data limit \cite{korda2018convergence}. In the context of ergodic dynamics with normal Koopman operators, the problem of approximating the spectral measure has been studied \cite{korda2020data, colbrook2024rigorous}. 

\par A major practical issue with the data-driven estimation of Koopman operator is that when the learning is based on snapshot data of the dynamics, the learned operator can fail to perform \emph{long-horizon predictions}. 
To mitigate such an error, in deep learning-based approaches, it has become a well-accepted approach to set the loss function for training as the sum of squared errors over a long horizon \cite{lusch2018deep, otto2019linearly}. In kernel methods, \cite{bevanda2024koopman} used a kernel function to account for the similarity between trajectories over a horizon (instead of snapshot points). These approaches are natural from a learning theory point of view -- to train on a horizon of $T$ for prediction on a horizon of $T$. 
On the other hand, when the nonlinear system is approximately embedded as such, the inherent qualitative behavior, specifically stability, still lacks a formal guarantee to be preserved. That is, \emph{a stable nonlinear dynamics can be learned as unstable with a small error on the Koopman operator}. 

\par To restrict the learned Koopman dynamics to be stable, Fan et al. \cite{fan2022learning} proposed to use a parameterization of Schur stable matrices, so that the learning of finite-rank approximation of Koopman operator can be cast as an unconstrained one. A similar parameterization for Hurwitz matrices (for continuous-time dynamics) was used in Bevanda et al. \cite{bevanda2022diffeomorphically}. 
Instead of imposing the stability constraint or using parameterizations, in this work, the notion of a \emph{space of weighted continuous functions}, namely the class of all functions in the form of $w(\cdot)g(\cdot)$ where $g$ is any continuous function and $w$ is a given weighting function, is introduced. Using a universal kernel $k(\cdot, \cdot)$, an associated \emph{weighted reproducing kernel Hilbert space} (RKHS) can be constructed to be dense in the space of weighted continuous functions. Here, the weighting function $w$ is one whose value decreases exponentially with the dynamics. 
We propose to learn the Koopman operator in such a weighted RKHS, which, if exactly learned, is guaranteed to be contractive. A probabilistic upper bound on the error of Koopman operator is derived in the similar form to \cite{kostic2022learning} but in a modified sense. 

\par In addition, here we consider the problem of \emph{estimating the Lyapunov function} (as a stability certificate and a prerequisite for Lyapunov-based control \cite{haddad2008nonlinear}) \emph{and Zubov function} (whose level sets give information about the domain of attraction (DoA) \cite{zubov1964methods}) of a globally and locally asymptotically stable nonlinear system, respectively. 
Evidently, to construct Lyapunov and Zubov functions, the learned Koopman (linear) dynamics needs to be a contractive one, so that the operator Lyapunov equation \cite{curtain1978infinite} can provide a unique solution that determines the Lyapunov function in the original nonlinear system. 
When the system is globally asymptotically stable, based on the error bound of learning the Koopman operator in the weighted RKHS, this paper derives a corresponding bound on the average estimation error of the Lyapunov function. 
When the system is only locally asymptotically stable, we learn the Zubov-Koopman operator, as defined in \cite{meng2023learning}, in the weighted RKHS, which then similarly guarantees an error-bounded estimation on the Zubov function. 

\par \emph{Notations:} In this paper, we use lower-case letters to represent scalars, vectors, and functions, and upper-case letters for matrices, operators, and sets. Upper letters in calligraphic fonts stand for function spaces, e.g., $\mathcal{C}(\Omega)$ -- space of continuous function defined on region $\Omega$, and $\mathcal{H}$ -- reproducing kernel Hilbert space, or $\mathcal{P}$ -- probability distribution. Vector norms, supremum norms of functions, and operator norms are denoted as $\|\cdot\|$. Inner products are denoted as $\langle \cdot, \cdot \rangle$. 
The product of two functions $w$ and $g$ in the pointwise sense is denoted as $w\cdot g$, i.e., $(w\cdot g)(x) = w(x)g(x)$. The composition of $g$ with $f$ is denoted as $g\circ f$, i.e., $(g\circ f)(x) = g(f(x))$, in particular, $f^2 := f\circ f$, and recursively $f^t:= f^{t-1} \circ f$ for $t=1, 2, \cdots$.
The set of Hilbert-Schmidt operators on $\mathcal{H}$ is denoted as $\mathrm{HS}(\mathcal{H})$ and the Hilbert-Schmidt norm of any $A\in \mathrm{HS}(\mathcal{H})$ is denoted as $\|A\|_{\mathrm{HS}}$. Finally, $\mathbb{I}_m = \{1, 2, \dots, m\}$.

\section{\textsc{Preliminaries}}
Consider a discrete-time nonlinear dynamical system (with time index $t\in \mathbb{N}$):
\begin{equation}\label{eq:system}
    x[t+1] = f(x[t])
\end{equation}
defined on a simply connected region $\Omega \subset \mathbb{R}^n$, i.e., $f: \Omega\rightarrow \Omega$. We aim to study the dynamics without modeling $f$ explicitly. 
\begin{assumption}\label{assum:dynamics}
    $\Omega$ is compact, and $f$ is continuous on $\Omega$. 
\end{assumption}

\subsection{Koopman Operator for Nonlinear Dynamics}
For the system \eqref{eq:system}, one can examine a corresponding ``flow'' with the dynamics of any function, called an \emph{observable}, defined on a function space $\mathcal{G}$. Suppose that the system starts from $x[0] = \xi \in \Omega$ and an observable $g\in\mathcal{G}$ is given. Then at $t=1, 2, \cdots$, the value of such an observable becomes $g(x[1]) = (g\circ f)(\xi)$, $g(x[2]) = (g\circ f\circ f)(\xi)$, $\cdots$. Thus, on the function space $\mathcal{G}$, the nonlinear dynamics is captured by the composition operator with $f$. 
\begin{definition}
    The \emph{Koopman operator} or \emph{composition operator} of system \eqref{eq:system} refers to $A: \mathcal{G}\rightarrow\mathcal{G}$, $g\mapsto g\circ f$. Here the function space $\mathcal{G}$ must be closed under $A$. By assuming $f\in \mathcal{C}(\Omega)$, here we use $\mathcal{G} = \mathcal{C}(\Omega)$. 
\end{definition}
Obviously, $A$ is a linear operator. Since $\|Ag\| = \sup_{x\in \Omega} |Ag(x)| = \sup_{x\in \Omega} |g(f(x))| \leq \sup_{y\in \Omega} |g(y)| = \|g\|$ holds for any $g\in \mathcal{C}(\Omega)$, we note that $A$ is a bounded operator with $\|A\|\leq 1$. Therefore, $A$ is a continuous and bounded operator on the Banach space $\mathcal{C}(\Omega)$. However, such a property may be too broad for data-based techniques. For the purpose of approximating the Koopman operator from a large but finite amount of data, RKHS is typically used in machine learning methods for nonlinear modeling, colloquially known as ``kernel tricks''.

\subsection{Reproducing Kernel Hilbert Space (RKHS)}
The definition of a RKHS involves a kernel function. Here we shall only consider real-valued ones. 
\begin{definition}
    A function $k:\Omega\times\Omega \rightarrow \mathbb{R}$ is a \emph{kernel} on $\Omega$ if for any finite choice of $x_1, \dots, x_m \in \Omega$, the resulting kernel matrix $K = [k(x_i, x_j)] \in \mathbb{R}^{m\times m}$ is symmetric and positive semidefinite. The RKHS on $\Omega$, denoted as $\mathcal{H} = \mathcal{H}_k(\Omega)$ is the completion of the following pre-Hilbert space:
    $$\mathcal{H}' = \left\{\sum_{i=1}^m c_ik(x_i, \cdot): x_i \in \Omega, c_i\in\mathbb{R}, \, i\in\mathbb{I}_m; m\in\mathbb{N} \right\}$$
    endowed with the inner product: $\langle k(x_i, \cdot), k(x_j, \cdot) \rangle = k(x_i, x_j)$. 
\end{definition}
We shall always assume that $k$ is a continuous function, so that any $k(x, \cdot)\in \mathcal{C}(\Omega)$ and thus $\mathcal{H}\subset \mathcal{C}(\Omega)$. 

\par Since $\Omega\subset\mathbb{R}^d$, the RKHS is a separable function space, i.e., it has a countable dense subset (e.g., one that contains all rational linear combinations of $k_x(\cdot)$ where $x\in \Omega \cap \mathbb{Q}^d$). Thus, $\mathcal{H}$ has a complete orthonormal basis. In fact, considering the integral operator defined by the kernel: 
$$T_k: \mathcal{H}\rightarrow \mathcal{H}, \, 
g\mapsto \int_\Omega k(\cdot, x)g(x)dx, $$ 
the pairs of eigenvalues and associated eigenfunctions can be arranged as $\{(\lambda_j, \phi_j)\}_{j=1}^\infty$ with $\lambda_1\geq \lambda_2 \geq \cdots \geq 0$. Hence, $\{(\lambda_j^{-1/2}\phi_j)\}_{j=1}^\infty$ forms an orthonormal basis for $\mathcal{H}$; that is, for any $g\in \mathcal{H}$, the sequence of its Fourier coefficients with respect to these basis functions becomes an element in $\ell^2$. These facts are well known in functional analysis, e.g., Lax \cite{lax2014functional}. Here it is desired that any continuous function can be sufficiently well approximated by a function in the RKHS. 
\begin{definition}
    A kernel $k$ on $\Omega$ is said to be universal if $\mathcal{H}_k(\Omega)$ is dense in $\mathcal{C}(\Omega)$.
\end{definition}
A sufficient condition for universality is given by the following fact \cite[Corollary~4.57]{steinwart2008support}. As examples, Gaussian kernels $k(x, x^\prime) = \exp(-\|x - x^\prime\|_2^2/h^2)$ ($h>0$) and Mat{\'{e}}rn kernels \cite{williams2006gaussian} are universal. 
\begin{fact}
    If the kernel $k$ can be expressed as 
    $k(x, x^\prime) = \psi(\langle x, x^\prime\rangle)$
    for an analytical function $\psi: (-r, r)\rightarrow \mathbb{R}$ with some $r>\max_{x\in \Omega} \|x\|_2$, then $k$ is a universal kernel. 
\end{fact}

\par When the chosen kernel is universal, it becomes natural to approximate the Koopman operator as a bounded operator on the RKHS. On a separable Hilbert space, a bounded operator can then be arbitrarily well approximated by Hilbert-Schmidt operators. Since $\mathrm{HS}(\Omega)$ is a Hilbert space \textit{per se}, the representer theorem will allow the convex learning of a Hilbert-Schmidt operator to admit a finite-rank solution \cite{khosravi2023representer}. However, it is noted that RKHS usually does not admit the definition of the composition operator, i.e., it is not warranted that $g\circ f \in \mathcal{H}$ for any $g\in \mathcal{H}$. Therefore, the approximate characterization of Koopman operator on RKHS needs to involve the injection from $\mathcal{H}$ to $\mathcal{C}(\Omega)$. This is the framework of Kostic et al. \cite{kostic2022learning}, which is reviewed next.

\subsection{Approximately Learning Koopman Operators on RKHS}\label{subsec:approximate}
\par Let us denote by $J: \mathcal{H} \rightarrow \mathcal{C}(\Omega), \, g\mapsto g$ as the injection map. To obtain an approximation of Koopman operator $A$ defined on $\mathcal{C}(\Omega)$ on its subspace $\mathcal{H}$, one should look for an operator $\hat{A}: \mathcal{H} \rightarrow \mathcal{H}$ such that $J\hat{A} \approx AJ: \mathcal{H} \rightarrow \mathcal{C}(\Omega)$ in a proper sense. Note that when $\Omega$ is a compact region and $k$ is a continuous kernel, the injection $J$ is a compact operator. Thus, $AJ$ is compact and then we look for $\hat{A}\in\mathrm{HS}(\mathcal{H})$ such that the following \emph{risk} is minimized:
\begin{equation}\label{eq:risk}
    \rho(\hat{A}) = \mathbb{E}_{x\sim \mathcal{P}} \left[\left\| \hat{A}^\ast k(x, \cdot) - k(f(x), \cdot)\right\|^2\right].
\end{equation}
Here the risk is evaluated as an expectation over a probability distribution $\mathcal{P}$ of interest on $\Omega$, under which $x$ is regarded as a random vector. $\hat{A}^\ast$ refers to the adjoint operator of $\hat{A}$. The reason to define $\rho$ based on the difference between $\hat{A}^\ast k(x, \cdot)$ and $k(f(x), \cdot)$ is that, if $\hat{A}$ was exactly learned (i.e., is the restriction of the actual Koopman operator $A$ restricted on $\mathcal{H}$), then it should satisfy, for any $x\in \Omega$ and $g\in \mathcal{H}$:
$$ \langle A^\ast k(x, \cdot), g\rangle = \langle k(x, \cdot), Ag\rangle = \langle k(x, \cdot), g(f(\cdot)) \rangle = g(f(x)), $$
namely $$A^\ast k(x, \cdot) = k(f(x), \cdot).$$

\par Suppose that an independent and identically distributed sample $\{(x_i, y_i)\}_{i=1}^m$ is drawn from $\mathcal{P}$, with $y_i=f(x_i)$. Denote by $\hat{\mathcal{P}} = \frac{1}{m}\sum_{i=1}^m \delta_{x_i}$ the empirical distribution. Then we write the \emph{empirical risk} as 
\begin{equation}\label{eq:empirical.risk}
    \begin{aligned}
        \hat{\rho}(\hat{A}) &= \mathbb{E}_{x\in \hat{\mathcal{P}}} \left[\left\| \hat{A}^\ast k(x, \cdot) - k(f(x), \cdot)\right\|^2\right] \\
        &= \frac{1}{m}\sum_{i=1}^m \left\| \hat{A}^\ast k(x_i, \cdot) - k(y_i, \cdot)\right\|^2 . 
    \end{aligned} 
\end{equation}
A natural idea to mitigate overfitting is to restrict the norm and rank of the learned operator $\hat{A}$. The \emph{reduced rank regression (RRR)} formulation for Koopman operator learning on RKHS is as follows:
\begin{equation}
    \hat{\rho}_{\beta, r} := \min_{\hat{A} \in \mathrm{HS}(\mathcal{H}), \, \mathrm{rank}\, \hat{A}\leq r} \hat{\rho}(\hat{A}) + \beta\|\hat{A}\|_{\mathrm{HS}}^2.
\end{equation}
According to the representer theorem \cite{khosravi2023representer}, the solution of the RRR problem can be restricted to $\mathrm{span}\{k(x_i, \cdot)\otimes k(y_j, \cdot): i, j\in \mathbb{I}_m\}$, where $\otimes$ is the tensor product. Specifically, $g_1\otimes g_2$ for given $g_1, g_2\in \mathcal{H}$ is the rank-$1$ operator such that $(g_1\otimes g_2)h = \langle g_2, h\rangle g_1$ (for any $h\in \mathcal{H}$). That is, 
\begin{equation}\label{eq:estimated-Koopman}
    \hat{A}_{\beta, r} = \sum_{i=1}^m \sum_{j=1}^m\theta_{ij}k(x_i, \cdot) \otimes k(y_j, \cdot)
\end{equation}
for some $\Theta = [\theta_{ij}] \in \mathbb{R}^{m\times m}$. In fact \cite{kostic2022learning}, the solution of $\Theta$ can be found explicitly by 
\begin{enumerate}
    \item calculating the kernel matrices $K = [k(x_i, x_j)]$ and $L = [k(y_i, y_j)]$, 
    \item solving the generalized eigenvalue problem $\frac{1}{m^2}LKu_i = \sigma_i^2 (\frac{1}{m}K+\beta I)u_i$ for $i\in \mathbb{I}_r$ corresponding to the largest $r$ eigenvalues $\sigma_1^2\geq \cdots \geq \sigma_r^2 \geq 0$, 
    \item normalizing the eigenvectors to $u_i^\top K(K+\beta I)u_i = 1$ for $i\in \mathbb{I}_r$; and finally, 
    \item letting $U_r = [u_1, \cdots, u_r]$ and $\Theta = (1/m)U_rU_r^\top K$.
\end{enumerate}

As $\|A\|\leq 1$, it is desired that $\hat{A}$ is also bounded. Here we provide only a rough upper bound on $\|\hat{A}\|$. 
\begin{proposition}
    The estimated Koopman operator $\hat{A}_{\beta, r}$ obtained above satisfies the following bound:
    \begin{equation}\label{eq:estimated-Koopman-bound}
        \|\hat{A}\|\leq \lambda_{\max}(L)/\beta m,  
    \end{equation}
    where $\lambda_{\max}(L)$ represents the largest eigenvalue of $L$. 
\end{proposition}
\begin{proof}
    From \eqref{eq:estimated-Koopman}, we calculate 
    $$\|\hat{A}\| = \lambda_{\max}(K^{-1/2}L\Theta^\top K\Theta LK^{-1/2}). $$
    Since $U_r^\top K(K+\beta I)U_r = I_r$, we have $U_r^\top KU_r \leq (1/\beta)I_r$. Thus, substituting the expression of $\Theta$ from step 4) and simplifying, we obtain
    \begin{align*}
        \|\hat{A}\|^2 &\leq \frac{1}{\beta m^2}\lambda_{\max} (U_r^\top (K+\beta I) \Sigma_r K^{-1} \Sigma_r (K+\beta I) U_r) \\
        &= \frac{1}{\beta m^2}\lambda_{\max} (U_r^\top KLK^{-1}LKU_r) \\
        &\leq \frac{\lambda^2_{\max}(L)}{\beta m^2}\lambda_{\max}(U_r^\top KU_r) \leq \frac{\lambda^2_{\max}(L)}{\beta^2 m^2}, 
    \end{align*}
    which is the proposition to prove. 
\end{proof}

\par In \cite{kostic2022learning}, a probabilistic bound on the generalization loss of the RRR estimator, i.e., $|\rho(\hat{A}_{\beta, r}) - \hat{\rho}_{\beta, r}|$, was derived in terms of the regularization coefficient $\beta$ and rank $r$ with confidence $1-\delta$. In the afore-mentioned paper, the dynamics was assumed to be measure-preserving and the true Koopman operator was defined on the $L_2$ space with respect to the invariant measure. In this work, we instead consider asymptotically stable dynamics. Similar bound will be derived and used for subsequent Koopman and Zubov function estimations.

\section{\textsc{Koopman and Zubov-Koopman Operators in Weighted Function Spaces}}
In order to preserve the asymptotic stability information in the Koopman operator estimation, we assume that we indeed have a ``certificate'' for asymptotic stability \textit{a priori}. 
\begin{assumption}\label{assum:weight}
    A function $w: \Omega \rightarrow [0, \infty)$ is known, which satisfies $w(0)=0$ and for all $x\in \Omega$, $w(f(x))/w(x) \leq \alpha$ for a known constant $\alpha \in (0, 1)$. 
\end{assumption}
Such a stability certificate stipulates the rate of attraction to the origin. For example, if $w(x) = \|x\|$, then the state converges to the origin at an exponential rate; if $w(x) = \exp(\|x\|^\alpha)$, then the state converges at the rate of $t^{-1/\alpha}$. It is reasonable to believe that the acquisition of function $w$ is much easier than modeling the nonlinear system itself, and can be based on very inaccurate models or even qualitative information about the dynamics. 

\par Within an RKHS framework, we also make the following standing assumption about the kernel to use.  
\begin{assumption}\label{assum:kernel}
    Kernel $k$ is continuous, bounded (without loss of generality, $\sup_{x\in \Omega} k(x,x) \leq 1$), and universal. 
\end{assumption}

\subsection{Koopman Operator in Weighted Function Spaces}
With the given $w$ function, we may now define the weighted continuous function space and weighted RKHS as follows. 
\begin{definition}
    The weighted continuous function space with weighting $w$ is 
    $$\mathcal{C}_w(\Omega) = \{w\cdot g: \, g\in \mathcal{C}(\Omega)\}, $$
    endowed with the norm: $\|w\cdot g\|_{\mathcal{C}_w(\Omega)} = \|g\|_{\mathcal{C}(\Omega)}$. 
    The weighted RKHS with kernel $k$ and weighting $w$ is  
    $$\mathcal{H}_{k,w}(\Omega) = \{w\cdot g: g\in \mathcal{H}_k(\Omega)\}, $$
    endowed with an inner product $\langle w\cdot g_1, w\cdot g_2 \rangle_{\mathcal{H}_{k,w}(\Omega)} = \langle g_1, g_2 \rangle_{\mathcal{H}_k(\Omega)}$.
\end{definition}
\begin{proposition}
    It is easy to observe the following properties. 
\begin{itemize}[nolistsep]
    \item $\mathcal{C}_w(\Omega)$ is a Banach space under its norm. 
    \item $\mathcal{H}_{k,w}(\Omega)$ is a RKHS by itself with the kernel defined by $k_w(x, y) := w(x)w(y)k(x, y)$ for all $x, y\in \Omega$. 
    \item $\mathcal{H}_{k,w}(\Omega)$ is dense in $\mathcal{C}_w(\Omega)$ with respect to the norm of $\mathcal{C}_w(\Omega)$ under Assumption \ref{assum:kernel}. 
\end{itemize} 
\end{proposition}

\par Now we may define the Koopman operator on the weighted continuous function space. 
\begin{definition}
    The \emph{Koopman operator} refers to the following linear operator 
    $$A: \mathcal{C}_w(\Omega)\rightarrow \mathcal{C}_w(\Omega), \, w\cdot g\mapsto (w\circ f) \cdot (g\circ f). $$
\end{definition}
\begin{proposition}
    Under Assumption \ref{assum:weight}, the Koopman operator
    $A: \mathcal{C}_w(\Omega)\rightarrow \mathcal{C}_w(\Omega)$ is contractive: $\|A\|\leq \alpha$. 
\end{proposition}
\begin{proof}
    By definition, for any $w\cdot g\in \mathcal{C}_w(\Omega)$
    \begin{align*}
        & \|A(w\cdot g)\|_{\mathcal{C}_w(\Omega)}
        = \|(w\circ f) \cdot (g\circ f)\|_{\mathcal{C}_w(\Omega)} \\
        &\quad = \sup_{x\in \Omega} \frac{w(f(x))|g(f(x))|}{w(x)} \leq \alpha\sup_{x\in \Omega} |g(f(x))| \\
        &\quad \leq \alpha \sup_{y\in\Omega} |g(y)| = \alpha \|w\cdot g\|_{\mathcal{C}_w(\Omega)}.
    \end{align*}
    That is the conclusion to be proved. 
\end{proof}
Therefore, as long as a ``stability certificate'' $w$ is obtained, by restricting the Koopman operator on the weighted continuous function space $\mathcal{C}_w(\Omega)$, its contraction is guaranteed. 
Since $\mathcal{H}_w(\Omega)$ is dense in $\mathcal{C}_w(\Omega)$, with a similar procedure to the one in \S\ref{subsec:approximate}, the Koopman operator can be approximated by a rank-$r$ operator on $\mathcal{H}_w(\Omega)$ from the dataset. Then the main problem of interest is whether such an estimation incurs a bounded error, so that the approximated Koopman operator is capable of predicting the future orbits.

\subsection{Learning Error of Koopman Operators in Weighted RKHS}
In the sequel, we use the following operators:
\begin{align*}
    X &= \mathbb{E}_{x\sim D} \left[ k_w(x, \cdot) \otimes k_w(x, \cdot) \right], \\
    \hat{X} &= \frac{1}{m}\sum_{i=1}^m k_w(x_i, \cdot) \otimes k_w(x_i, \cdot) \\
    Y &= \mathbb{E}_{x\sim D} \left[ k_w(f(x), \cdot) \otimes k_w(f(x), \cdot)\right], \\
    \hat{Y} &= \frac{1}{m}\sum_{i=1}^m k_w(y_i, \cdot) \otimes k_w(y_i, \cdot) \\
    Z &= \mathbb{E}_{x\sim D} \left[ k_w(x, \cdot) \otimes k_w(f(x), \cdot)\right], \\
    \hat{Z} &= \frac{1}{m}\sum_{i=1}^m k_w(x_i, \cdot) \otimes k_w(y_i, \cdot),  
\end{align*}
based on which we can rewrite, for $G \in \mathrm{HS}(\mathcal{H}_w)$: 
\begin{equation}\label{eq:generalization}
\begin{aligned}
    \rho_w(\hat{A}) - \hat{\rho}_w(\hat{A}) = \mathrm{tr}\big[&(Y - \hat{Y}) + \hat{A}\hat{A}^\ast(X - \hat{X}) \\
    & - \hat{A}^\ast(Z-\hat{Z}) - (Z-\hat{Z})^\ast \hat{A} \big], 
\end{aligned}
\end{equation}
where $\rho_w$ and $\hat{\rho}_w$ refer to the risk and empirical risk defined in \eqref{eq:risk} and \eqref{eq:empirical.risk}, respectively, with $k$ replaced by $k_w$. 
We note the following lemma from Pontil and Maurer \cite{pontil2013excess} on the concentration of the average of rank-$1$ operators around the expectation. The subsequent lemma is the well-known Bernstein's inequality from probability theory. With these, the generalization loss (namely the difference between the risk and empirical risk) becomes probabilistically bounded. 
\begin{lemma}\label{lemma:Pontil}
    Suppose that $T_1, \cdots, T_m$ are independently distributed random rank-1 operators, satisfying $\|\mathbb{E}[T_1]\|\leq 1$. Then for $\bar{T} := \frac{1}{m}\sum_{i=1}^m T_i$, for any $\epsilon > 0$: 
    $$\mathbb{P}\left[\|\bar{T} - \mathbb{E}[T_1]\| > \epsilon\right] \leq 4m^2\exp\left(-\frac{m\epsilon^2}{9 + 6\epsilon}\right).$$
\end{lemma}
\begin{lemma}\label{lemma:Bernstein}
    Suppose that $\xi_1, \cdots, \xi_m$ are independently distributed random variables, such that $|\xi_1|\leq 1$ almost surely. Then for $\bar{\xi} := \frac{1}{m}\sum_{i=1}^m \xi_i$, for any $\epsilon > 0$: 
    $$\mathbb{P}\left[|\bar{\xi} - \mathbb{E}[\xi_1]| > \epsilon\right] \leq 2\exp\left(-\frac{m\epsilon^2/2}{\mathbb{E}[\xi_1^2] + \epsilon/3}\right).$$
\end{lemma}

\begin{theorem}\label{th:generalization}
    Under Assumptions \ref{assum:dynamics}, \ref{assum:weight}, and \ref{assum:kernel}, with independent sample $\{(x_i, y_i = f(x_i))\}_{i=1}^m$ from probability distribution $\mathcal{P}$ supported on $\Omega$, the following inequality holds with probability at least $1-\delta$:
    \begin{equation}\label{eq:generation-bound}
        \rho_w(\hat{A}) \leq \hat{\rho}_w(\hat{A}) + \check{\rho}_w(m, \gamma, r, \delta) 
    \end{equation}
    in which
    \begin{align*}
        & \check{\rho}_w(m, \gamma, r, \delta):= \frac{1}{m}\log\frac{6}{\delta} + \sqrt{\frac{8}{m}\log \frac{6}{\delta}} \\
        & + \gamma\left(\gamma+2\sqrt{r} \right)\left(\frac{6}{m}\log \frac{12m^2}{\delta} + \sqrt{\frac{9}{m}\log\frac{12m^2}{\delta}}\right),      
    \end{align*}
    for any $\hat{A}\in \mathrm{HS}(\mathcal{H}_w)$ such that $\|\hat{A}\|_{\mathrm{HS}}\leq \gamma$ and $\mathrm{rank}\, \hat{A}\leq r$. 
\end{theorem}
\begin{proof}
    The right-hand side of \eqref{eq:generalization} can be relaxed to $\mathrm{tr}[Y - \hat{Y}] + \gamma^2 \|X - \hat{X}\| + 2\sqrt{r}\gamma \|Z - \hat{Z}\|$. By Lemma \ref{lemma:Pontil}, to ensure the probability of
    $\|X - \hat{X}\| > \epsilon_X$ not to exceed $\delta/3$, we need 
    $$4m^2 \exp\left(-\frac{m\epsilon_X^2}{9 + 6\epsilon_X}\right) \leq \frac{\delta}{3}.$$
    Solving the above inequality, an overestimation is
    $$\epsilon_X = \frac{6}{m}\log \frac{12m^2}{\delta} + 3\sqrt{\frac{1}{m}\log\frac{12m^2}{\delta}}. $$
    Similarly, the probability of $\|Z - \hat{Z}\|>\epsilon_Z$ does not exceed $\delta/3$ for $\epsilon_Z = \epsilon_X$. 
    Then, according to Lemma \ref{lemma:Bernstein}, in order that the probability of $\mathrm{tr}[Y - \hat{Y}] > \epsilon_Y$ does not exceed $\delta/3$, we need
    $$2\exp\left(-\frac{m\epsilon_Y^2/2}{1 + \epsilon_Y/3}\right) \leq \frac{\delta}{3}.$$
    An overestimation is
    $$\epsilon_Y = \frac{1}{m}\log\frac{6}{\delta} + \sqrt{\frac{8}{m}\log \frac{6}{\delta}}.$$
    The conclusion is thus proved. 
\end{proof}
A direct implication of such a bound in generalization loss is a bound on expected prediction error when applying $\hat{A}$ on any observable on $\mathcal{C}_w(\Omega)$. 
\begin{corollary}
    Under the same conditions as in Theorem \ref{th:generalization}, for any $g\in \mathcal{C}(\Omega)$, with probability at least $1-\delta$:
    \begin{equation}
    \begin{aligned}
        \mathbb{E}_{x\sim \mathcal{P}}\left[\big{|} \hat{A}(w\cdot g)(x) - w(f(x))g(f(x)) \big{|}^2 \right] \leq \\ \left(\hat{\rho}_w(\hat{A}) + \check{\rho}_w(m, \gamma, r, \delta)\right) \|g\|_{C(\Omega)}^2.
    \end{aligned}
    \end{equation}
\end{corollary}

\subsection{Extension to Zubov-Koopman Operators}
In Meng et al. \cite{meng2023learning}, the concept of Zubov-Koopman operator was proposed to characterize the DoA when it does not contain the entire $\Omega$. The definition (for a discrete-time system) was given as $Z: g\mapsto e^{-\eta(\cdot)}g(f(\cdot))$ where $\eta$ is a positive definite function that assigns a cost to any state on $\Omega$. 
Thus, $Z^t: g\mapsto \exp\left(-\sum_{s=0}^{t-1} \cdot \eta\circ f^s\right) \cdot (g\circ f^t)$. It was pointed out in \cite{meng2023learning} that by letting $g$ have a value equal to $1$ at the origin (e.g., $g\equiv 1$ and $t\rightarrow\infty$), $Z^\infty g = \exp(-\sum_{s=1}^\infty \eta\circ f^s(\cdot)) =: \zeta$ becomes a \emph{Zubov function}, which has a support on the DoA, since all points not in DoA will have an infinite accumulated cost on the exponent. (Strictly speaking, $1-\zeta$ is the actual ``Zubov-Lyapunov function'' whose value is below $1$ on the DoA, while $\zeta$ can be called as a ``Zubov dual function''. For simplicity, we refer to $\zeta$ as the Zubov function.)
However, given a dataset $\{(x_i,y_i)\}_{i=1}^m$, the estimated $\hat{Z}$ contains error and does not guarantee that $\hat{Z}^\infty$ approximates $Z^\infty$ well. 
Therefore, here we define the Zubov-Koopman operator on $\mathcal{C}_w(\Omega)$ with a weighting function $w: \Omega\rightarrow \mathbb{R}_+$. 
\begin{definition}
    The Zubov-Koopman operator $Z: \mathcal{C}_w(\Omega) \rightarrow \mathcal{C}_w(\Omega)$ is such that $w\cdot g\mapsto e^{-\eta(\cdot)}\cdot(w\circ f)\cdot (g\circ f)$. 
\end{definition}
Since the DoA is an open set, it cannot be guaranteed that $w(f(x))/w(x)\leq \alpha$ always hold for an $\alpha\in(0, 1)$. Instead, we assume that it is possible to find an $\eta$ that compensates the locality of DoA. 
\begin{proposition}\label{prop:Zubov-bound}
    Suppose that $\eta$ is chosen as such that for some $\alpha\in (0, 1)$, $e^{-\eta(x)} w(f(x)) \leq \alpha w(x)$ holds for all $x$. Then $Z: \mathcal{C}_w(\Omega) \rightarrow \mathcal{C}_w(\Omega)$ satisfies $\|Z\|\leq \alpha$. 
\end{proposition}

\par It is not hard to find that the dual operator of the Zubov-Koopman operator satisfies 
$$Z^\ast k_w(x, \cdot) = e^{-\eta(x)} k_w(f(x), \cdot).$$
Hence in a similar way, with an independent sample $\{(x_i,y_i=f(x_i))\}_{i=1}^m$ from probability distribution $\mathcal{P}$, the Zubov-Koopman operator is estimated on the weighted RKHS by minimizing the following empirical loss:
\begin{equation}
    \hat{\rho}_w(\hat{Z}) = \frac{1}{m} \sum_{i=1}^m \left\| Z^\ast k_w(x, \cdot) - e^{-\eta(x)} k_w(f(x), \cdot) \right\|^2. 
\end{equation}
The RRR formulation described in \S\ref{subsec:approximate} can be adopted by only modifying the definition of $L = [e^{-\eta(x_i)}e^{-\eta(x_j)}k_w(y_i, y_j)]$. The same bound on the risk can be derived as in Theorem \ref{th:generalization}. 

\par Next, with bounded-error estimations of the Koopman and Zubov-Koopman operators, we proceed with the estimations of Lyapunov and Zubov functions, respectively, for globally and locally asymtotically stable systems.

\section{\textsc{Estimation of Lyapunov and Zubov Functions}}
\subsection{Estimation of Lyapunov Function}
Suppose that the system \eqref{eq:system} is globally asymptotically stable (i.e., that the DoA contains the entire $\Omega$), then a Lyapunov function can be found as a certificate of attraction. Based on the estimated Koopman operator $\hat{A}$ on the weighted RKHS, a Lyapunov function can be estimated in a ``kernel-quadratic form'' as 
\begin{equation}
    \hat{v}(x) = \langle k_w(x,\cdot), \hat{P}k_w(x, \cdot)\rangle, \enskip \hat{A}\hat{P}\hat{A}^\ast - \hat{P} = -Q, 
\end{equation}
in which $\hat{P}$ is self-adjoint and $Q$ is a given bounded, self-adjoint, and positive operator. In order that the solution to the above operator Lyapunov equation uniquely exists, we shall assume that $\hat{A}$ is contractive on the weighted RKHS $\mathcal{H}_w$ in a similar manner to the true Koopman operator on $\mathcal{C}_w$. If so, then we have 
\begin{equation}\label{eq:Lyapunov-operator-estimated}
    \hat{P} = \sum_{t=0} \hat{A}^t Q(\hat{A}^\ast)^t. 
\end{equation}

\begin{assumption}\label{assum:Koopman-bound}
    The estimated Koopman operator $\hat{A}$ is such that $\|\hat{A}\|<1$. 
\end{assumption}
In view of \eqref{eq:estimated-Koopman-bound}, an upper bound of $\|\hat{A}\|$ is $\lambda_{\max}(L)/m\beta$. It is known that as $m\rightarrow\infty$, $\frac{1}{m}\lambda_{\max}(L)$ converges in probability to the largest eigenvalue of an integral operator \cite{koltchinskii2000random, braun2006accurate}. Hence, the assumption above is satisfied with high probability when $\beta$ is large, although, given the conservativeness of \eqref{eq:estimated-Koopman-bound}, $\beta$ does not necessarily need to be as large as the inequality suggests. 

The \emph{actual} Lyapunov function $v$ that satisfies 
$$v(f(x)) - v(x) = \langle k_w(x, \cdot), Qk_w(x, \cdot)\rangle, $$
however, should be 
$$ v(x) = \sum_{t=0}^\infty \langle k_w(x[t], \cdot), Qk_w(x[t], \cdot)\rangle, \enskip x[0] = x.$$
Since $k_w(x[t], \cdot) = (A^\ast)^t k_w(x, \cdot)$, we have 
\begin{equation}\label{eq:Lyapunov-operator}
    v(x) = \langle k_w(x, \cdot), Pk_w(x, \cdot)\rangle, \enskip P = \sum_{t=0}^\infty A^t Q (A^\ast)^t .
\end{equation}
Here due to the contractiveness of $A$, $P$ becomes a bounded, self-adjoint, and positive operator on $\mathcal{C}_w$ as the completion of $\mathcal{H}_w$. Its definition can be naturally extended from $\mathcal{H}_w$. 
\begin{theorem}\label{th:Lyapunov}
    Under the afore-mentioned assumptions, 
    \begin{equation}\label{eq:Lyapunov-bound}
        \mathbb{E}_{x\sim \mathcal{P}} \left[|\hat{v}(x)-v(x)| \right] \leq \frac{2\bar{\alpha} \|Q\|}{(1-\bar{\alpha}^2)^2} \rho_w(\hat{A})^{1/2}, 
    \end{equation}
    in which $\bar{\alpha}:=\max\{\|\hat{A}\|, \|A\|\}$. 
\end{theorem}
\begin{proof}
    Note that for any $t\in \mathbb{N}\backslash\{0\}$:
    $$\hat{A}^t Q(\hat{A}^\ast)^t - A^t Q(A^\ast)^t = (\hat{A} - A) R_t + R_t^\ast (\hat{A} - A)^\ast,$$
    where $R_t$ is a polynomial of $\hat{A}$ and $A$, with $\|R_t\| \leq t \|Q\|\bar{\alpha}^{2t-1}$. Thus, in view of \eqref{eq:Lyapunov-operator-estimated} and \eqref{eq:Lyapunov-operator}, 
    $$\hat{P}-P = (\hat{A} - A) R + R^\ast (\hat{A} - A)^\ast, $$
    where $R :=\sum_{t=1}^\infty R_t$ satisfies 
    $$\|R\| \leq \|Q\|\sum_{t=1}^\infty t \bar{\alpha}^{2t-1} = \frac{\bar{\alpha}\|Q\|}{(1-\bar{\alpha}^2)^2}. $$ 
    Thus, 
    \begin{align*}
    \text{L.H.S. of }\eqref{eq:Lyapunov-bound} & =  \mathbb{E}_{x\sim \mathcal{P}}\left[\big{|} 2\langle (\hat{A} - A)^\ast k_w(x, \cdot), Rk_w(x, \cdot)\rangle \big{|}\right] \\
    & \leq \frac{2\bar{\alpha}\|Q\|}{(1-\bar{\alpha}^2)^2} \mathbb{E}_{x\sim \mathcal{P}}\left[ \| (\hat{A} - A)^\ast k_w(x, \cdot) \| \right].
    \end{align*}
    By Cauchy-Schwarz inequality, the conclusion is proved. 
\end{proof}
According to Theorem \ref{th:Lyapunov}, with a sample size $m$ and confidence $1-\delta$, the mean-squared estimation error of the Lyapunov function is probabilistic bounded by $2\bar{\alpha}(1-\bar{\alpha}^2)^{-2}\|Q\|$ multiplied by the right-hand side of \eqref{eq:generation-bound}. 

\par Here we examine the practical performance of such an approximation scheme with a numerical case study. 

\begin{example}
    We consider the following system \cite{tang2024data}:
    \begin{equation}
        \dot{x}_1 = -3x_1 + x_2 + \frac{1}{2\pi}\sin 2\pi x_1, \enskip 
        \dot{x}_2 = x_1 - x_2. 
    \end{equation}
    The system has a globally stable equilibrium point at the origin (which can be shown using quadratic Lyapunov function if the model is known). Without an available model, assuming that the knowledge about the exponential decay of $w(\cdot) = \|\cdot\|$ is known, we sample $m=500$ points under a uniform distribution over $\Omega = \{x\in\mathbb{R}^2: \|x\|\leq 2\}$ with a discretization time of $0.05$, in order to learn the Koopman operator on $\mathcal{H}_{k, w}$ with Gaussian kernel $k(x, y) = \exp(-4\|x - y\|^2)$. In the learning, the regularization $\beta$ is chosen to be $1/100$ of the largest eigenvalue of $\frac{1}{m}K$, and the rank of $\hat{A}$ is chosen as $r=50$. The proposed method is then used to construct the Lyapunov function with $Q$ being the identity map. 
    \begin{figure}
        \centering
        \includegraphics[width=\columnwidth]{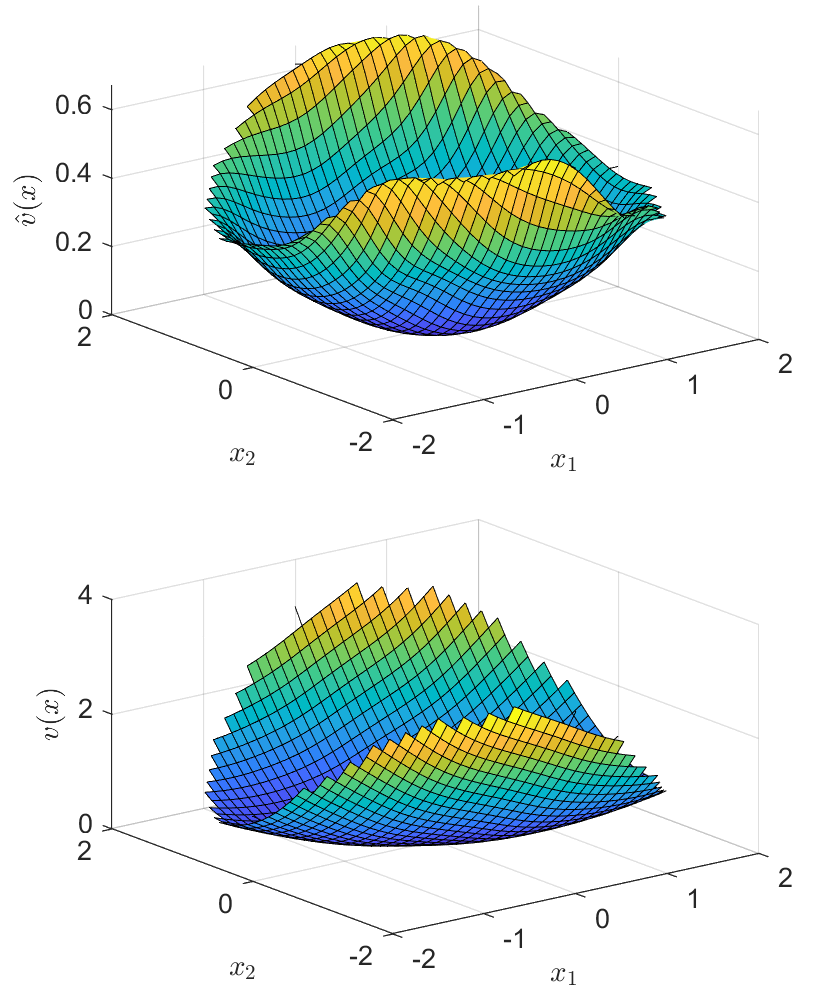}
        \caption{Prediction of Lyapunov function by the Koopman operator learned on weighted RKHS.}
        \label{fig:Ex1-Lyapunov}
    \end{figure}
    \par In this way, we constructed the Lyapunov function estimate $\hat{v}$. A plot is given in Fig. \ref{fig:Ex1-Lyapunov} in comparison to the actual Lyapunov function $V$. It is seen that the estimated Lyapunov function recovers a similar near-ellipsoidal contour shape as $v$. On the other hand, the value of $\hat{v}(x)$ grows less steep than $v(x)$, when $x$ is away from the origin, leading to a visually discernible deformation in the surface plot of the estimated Lyapunov function. 
    Essentially, the learning of a Koopman operator on a weighted RKHS favors the prediction of the dynamics on higher-weight regions (when $\|x\|$ is large) than lower-weight regions (when $\|x\|$ is small). As shown in Fig. \ref{fig:Ex1-prediction} that plots the actual v.s. predicted evolution of the values of two quadratic functions (starting from a random state), both predictions align well with the dynamics before $x$ is close to the origin. The long-term contributions to the Lyapunov function is therefore clearly under-estimated. 
    \begin{figure}
        \centering
        \includegraphics[width=\columnwidth]{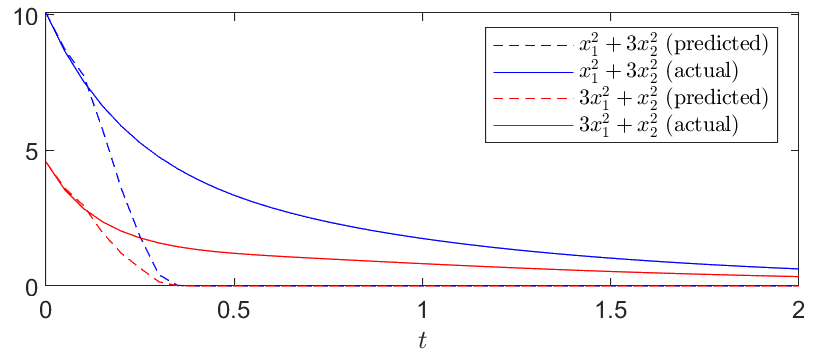}
        \caption{Prediction of quadratic functions by the Koopman operator learned on weighted RKHS.}
        \label{fig:Ex1-prediction}
    \end{figure}
\end{example}

\subsection{Estimation of Zubov Function and DoA}
\par For the estimation of Zubov function, we assume that $\hat{Z}$ and $Z$ are both contractive.
\begin{assumption}\label{assum:Zubov-Koopman-bound}
    $\bar{\alpha} := \max\{\|Z\|, \|\hat{Z}\|\} \in (0, 1)$. 
\end{assumption}
However, the Zubov function should not be estimated as $\hat{Z}^\infty g$ with $g\equiv 1$, since the constant function $1$ does not belong to $\mathcal{C}_w(\Omega)$. Instead, we consider 
\begin{equation}\label{eq:Zubov-observable}
    w(x)g(x) = \frac{w(x)^\nu}{w(x)^\nu + \varsigma^\nu}, \enskip x\in\Omega.
\end{equation}
for given (small) $\varsigma>0$ and some $\nu \geq 1$. This function appears like $1$ in the region where $w(x)\gg \varsigma$ and like $(w(x)/\varsigma)^\nu$ when $x$ is close to the origin. Its $\mathcal{C}_w$-norm is then bounded by 
$$c_\nu = \begin{cases}
    \varsigma^{-1}, & \text{if } \nu=1 \\
    \nu^{-1}(\nu-1)^{(\nu-1)/\nu}\varsigma^{-1}, & \text{if } \nu>1
\end{cases}.$$
With this $w\cdot g$ function as the observable, applying the actual and estimated Zubov-Koopman operator for time $t$, we obtain an approximate Zubov function and its data-driven estimate: 
\begin{equation}\label{eq:Zubov-estimate}
    \hat{\zeta}^t(x) = \hat{Z}^t(w\cdot g), \enskip \zeta^t(x) = Z^t(w\cdot g).
\end{equation}

\par First, we assert that $\zeta^t$ generated under the actual Zubov-Koopman operator is a ``good'' indication of the DoA. To this end, some regular assumptions are made on the pertinent functions $w$ and $\eta$. 
\begin{proposition}\label{prop:Zubov}
    Assume that 
    \begin{enumerate}[label=(\roman*)]
        \item for all $x$ not in the DoA, $\eta(x)\geq \underline{\eta}$ for some $\underline{\eta}>0$;
        \item there exists a sublevel set of the weighting function $S_a = \{x\in\Omega: w(x)\leq a\}$ that is contained in the DoA, and $\mu_a := \sup_{x\in S_a} \sum_{t=0}^\infty (\eta\circ f^t)(x) < \infty$; 
        \item there exists a constant $\underline{\alpha}>0$ such that $w(f(x)) \geq \underline{\alpha} w(x)$ for all $x$ in the DoA. 
    \end{enumerate}
    Then for all $x$ not in the DoA, $\zeta^t(x) \leq e^{-\underline{\eta}t}$; for all $x$ such that $w(x) = a$ and $t\leq \log_{1/\underline{\alpha}}(a/\varsigma)$, $\zeta^t(x)\geq e^{-\mu_a}/2$. 
\end{proposition}
\begin{proof}
    We note that 
    $$\zeta^t(x) = \exp\left(\sum_{s=0}^{t-1} \eta\circ f^s (x) \right) (w\cdot g)\circ f^t(x).$$
    For $x$ not in the DoA, the exponential factor in $\zeta^t$ is upper-bounded by $\exp(\underline{\eta}t)$ by condition (i), while the remaining term is bounded by $1$. Thus the first half of the proposition is proved. For any $x$ such that $w(x) = a$ (and thus contained in the DoA), the exponential factor in $\zeta^t$ is lower-bounded by $\exp(-\mu)$ by condition (ii), and the remaining term is lower-bounded by $(\underline{\alpha}^t w(x))^\nu / [(\underline{\alpha}^t w(x))^\nu + \varsigma^\nu]$. When $t$ is short such that $\underline{\alpha}^t a \geq \varsigma$ (i.e., $t\leq \log_{1/\underline{\alpha}} (a/\varsigma)$), we have $\zeta^t(x) \geq e^{-\mu_a}/2$. 
\end{proof}

\begin{corollary}
    Suppose that the assumptions in Proposition \ref{prop:Zubov} hold. Then the DoA contains $S_a = \{x:w(x)\leq a\}$ where 
    $$a = \sup\left\{a: \log\frac{\underline{\alpha}a}{\varsigma} \geq \frac{\mu_a + \log 2}{\underline{\eta}}\log\frac{1}{\underline{\alpha}} \right\}. $$
\end{corollary}
\begin{proof}
    Based on the previous proposition, when $t = \lfloor \log_{1/\underline{\alpha}} (a/\varsigma) \rfloor$, $\zeta^t(x) \geq e^{-\mu_a}/2$ for all $x\in S_a$ if $S_a$ is contained in the DoA, and $\zeta^t(x) \leq \exp(-\underline{\eta} \lfloor \log_{1/\underline{\alpha}} (a/\varsigma) \rfloor) \leq \exp(-\underline{\eta}(\log_{1/\underline{\alpha}} (a/\varsigma) - 1)$ for all $x$ not in the DoA. For $S_a$ to be contained in the DoA, it suffices to have 
    $$\underline{\eta}(\log_{1/\underline{\alpha}} (a/\varsigma) - 1) \geq \mu_a + \log 2, $$
    which is the condition in the corollary. 
\end{proof}

Finally, similar to the Lyapunov function estimation, we establish a mean squared error bound on the estimation $\hat{\zeta}^t$. 
\begin{theorem}
    Suppose that Assumptions \ref{assum:dynamics}, \ref{assum:kernel}, and \ref{assum:Zubov-Koopman-bound} hold. Then given \eqref{eq:Zubov-observable} and \eqref{eq:Zubov-estimate}, for all $t \in \mathbb{N}\backslash \{0\}$, we have
    \begin{equation}\label{eq:Zubov-bound}
        \mathbb{E}_{x\sim\mathcal{P}}\left[ \big{|}\hat{\zeta}^t(x) - \zeta^t(x)\big{|} \right] \leq t\bar{\alpha}^{t-1} \rho_w(\hat{Z})^{1/2}c_\nu \varsigma^{-1}, 
    \end{equation}
    where 
    $$\rho_w(\hat{Z}) = \mathbb{E}_{x\sim \mathcal{P}}\left[\|\hat{Z}^\ast k_w(x, \cdot) - Z^\ast k_w(x, \cdot)\|^2 \right].$$
\end{theorem}
\begin{proof}
    Note that $\hat{\zeta}(x) - \zeta(x) = \langle k_w(x, \cdot), (\hat{Z}^t-Z^t)(w\cdot g)\rangle = \langle ((\hat{Z}^\ast)^t - (Z^\ast)^t)k_w(x,\cdot), w\cdot g\rangle$. Thus
    $$\text{L.H.S. of \eqref{eq:Zubov-bound}} \leq t\bar{\alpha}^{t-1} \mathbb{E}[\|(\hat{Z} - Z)^\ast k_w(x, \cdot)\|]\|w\cdot g\|_{\mathcal{C}_w}.$$
    Using Cauchy-Schwarz inequality:
    $$\mathbb{E}[\|(\hat{Z} - Z)^\ast k_w(x, \cdot)\|] \leq \mathbb{E}[\|(\hat{Z} - Z)^\ast k_w(x, \cdot)\|^2]^{1/2}, $$
    the conclusion is obtained. 
\end{proof}

\begin{example}
    We consider the following system \cite[Ex.~4.28]{khalil2002nonlinear}:
    \begin{equation}
        \dot{x}_1 = -x_1, \enskip 
        \dot{x}_2 = (x_1x_2 - 1)x_2^3 + (x_1x_2-1+x_1^2)x_2. 
    \end{equation}
    The system has only one equilibrium point at the origin. Yet since $R = \{(x_1, x_2): x_1x_2 \geq 2\}$ is an invariant set, the origin cannot be globally asymptotically stable. We therefore would like to estimate the Zubov-Koopman operator on $\mathcal{H}_{k, w}(\Omega)$ with $\Omega = [-2, 2]\times [-2, 2]$, $w(x) = \|x\|^{1/2}$, $\eta(x) = \|x\|^2/2$ and $k(x, y) = \exp(-4\|x - y\|^2)$. We sampled $m=500$ points under the uniform distribution over $\Omega$ with a discretization time of $0.025$. Due to the same reason as explained in the previous example, the dynamics predicted by the Zubov-Koopman operator in the region close to the origin tends to have an overestimated contraction. Hence, in the plot of $\hat{\zeta}^t$, there will be an extra low-value region near the origin even though it is contained in the DoA. As shown in Fig. \ref{fig:Ex2-Zubov}, the surface plot of the estimated Zubov function $\hat{\zeta}^t$ (using $t=0.15$) has a lowered central region (encircled by a high-valued region). However, similar to the actual Zubov function $\zeta$, the estimated $\hat{\zeta}$ has close-to-zero values in the invariant, non-attracting region $R$.  
    \begin{figure}
        \centering
        \includegraphics[width=\columnwidth]{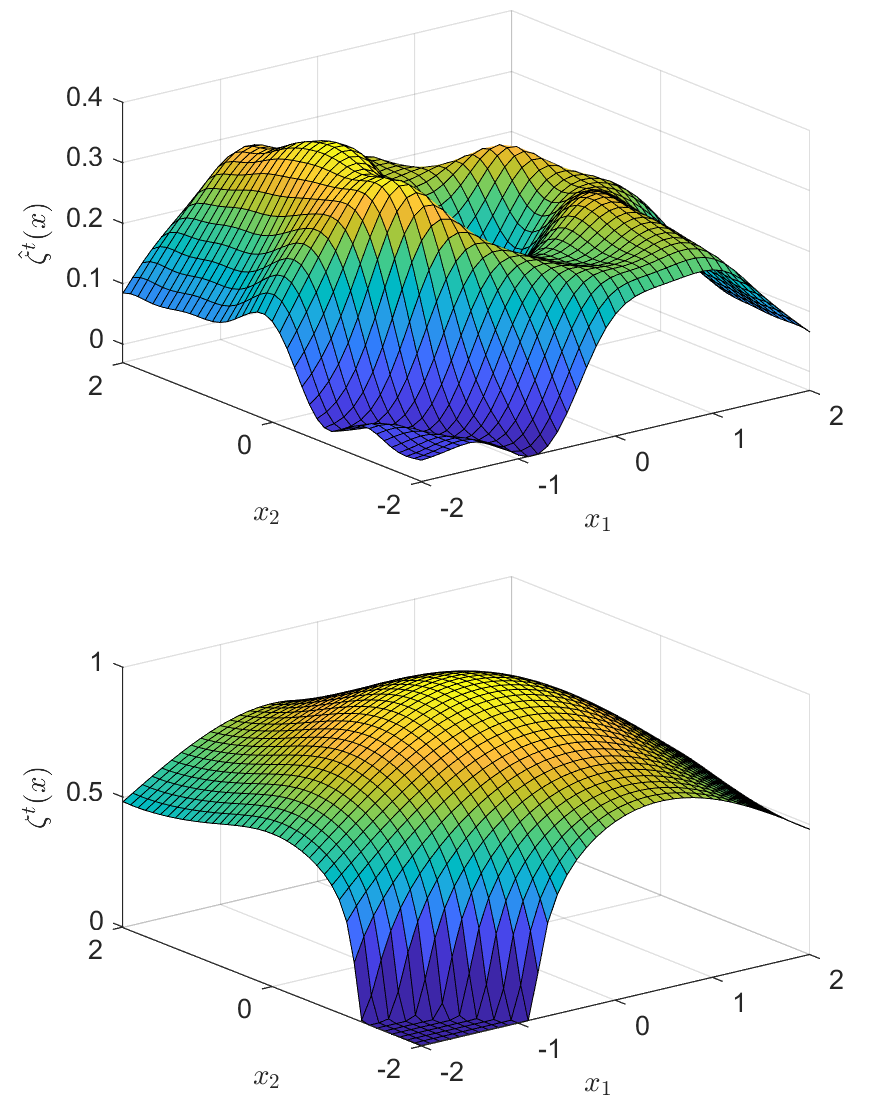}
        \caption{Prediction of the Zubov function by the Zubov-Koopman operator learned on weighted RKHS.}
        \label{fig:Ex2-Zubov}
    \end{figure}
\end{example}

\section{\textsc{Conclusions}}
In this paper, the concepts of weighted continuous function space and weighted RKHS are introduced for kernel-based learning of Koopman operators of asymptotically stable nonlinear systems, where the weighting function represents the prior knowledge about the decay rate. A probabilistic bound on the generalization loss (i.e., the $L_2$ error on predicting the evolution of kernel functions) is provided, which results in probabilistically correct constructions of Lyapunov functions and Zubov functions for DoA estimation. 
Intrinsically, the proposed approach exploits the prior information that the system has a contractive weighting function $w$, or so with the compensation of a given contraction factor $e^{-\eta(\cdot)}$. Hence, a Koopman description of the unknown dynamics is learned in a way that favors the prediction accuracy within high-weight regions. Therefore, despite the error in the low-weight state prediction, the stability-related information is preserved. 

\par For controlled systems, designing a control law for closed-loop stability with a specified Lyapunov function or DoA are challenging problems in a Koopman framework. The investigations will be carried out in future works.  

\bibliographystyle{ieeetr}
\bibliography{root.bib}
\end{document}